\documentclass[runningheads]{llncs}

\usepackage[T1]{fontenc}
\usepackage{graphicx}
\usepackage{fullpage}
\usepackage[sort,nocompress]{cite}

\usepackage{amsmath}

\usepackage{mathrsfs} 
\usepackage{xspace}
\usepackage{diagbox}
\usepackage{amsfonts}

\usepackage{enumitem}
\usepackage{nicefrac}
\usepackage{comment}

\usepackage[textsize=small]{todonotes}

\usepackage{amsmath}
\usepackage{mathtools}

\usepackage{xcolor,colortbl}
\usepackage{tabularx}
\usepackage{multirow}
\usepackage{booktabs}
\usepackage{array}
\usepackage{tcolorbox}

\usepackage{rotating}

\newcommand{\defprobgoal}[3]{
\begin{tcolorbox}[colback=gray!5!white,colframe=gray!75!black]
  \begin{tabular*}{\textwidth}{@{\extracolsep{\fill}}lr} #1   \\ \end{tabular*}
  {\bf{Input:}} #2  \\
  {\bf{Goal:}} #3
  \end{tcolorbox}
}

\newcommand{\defprob}[3]{
\begin{tcolorbox}[colback=gray!5!white,colframe=gray!75!black]
  \begin{tabular*}{\textwidth}{@{\extracolsep{\fill}}lr} #1   \\ \end{tabular*}
  {\bf{Input:}} #2  \\
  {\bf{Question:}} #3
  \end{tcolorbox}
}

\newcommand{\defparprob}[4]{
\begin{tcolorbox}[colback=gray!5!white,colframe=gray!75!black]
  \begin{tabular*}{\textwidth}{@{\extracolsep{\fill}}lr} #1   \\ \end{tabular*}
  {\bf{Input:}} #2  \\
  {\bf{Parameter:}} #3  \\  
  {\bf{Question:}} #4
  \end{tcolorbox}
}

\newcommand{\partition}{3-{\sc Partition}\xspace}

\newcommand{\NumMatch}{{\sc Numerical 3-Dimensional Matching}\xspace}
\newcommand{\cffafull}{{\sc Conflict free Fair Allocation}\xspace}
\newcommand{\cffa}{{\sc CFFA}\xspace}

\newcommand{\dcffafull}{{\sc Size bounded-Conflict free Fair Allocation}\xspace}
\newcommand{\dcffa}{{\sc Sb-CFFA}\xspace}

\newcommand{\ptime}{|{\mathscr J}|^{\Oh(1)}\xspace}

\newcommand{\util}{\ensuremath{{\rm u}}}

\newcommand{\cffag}{{\sc ${\mathcal G}$-\dcffa}\xspace}

\newcommand{\mwis}{{\sc Sb-MWIS}\xspace}

\newcommand{\sbmwisfull}{{\sc Size bounded-Maximum Weight Independent Set}\xspace}
\newcommand{\sbmwis}{{\sc Sb-MWIS}\xspace}

\newcommand{\mwisg}{{\sc ${\mathcal G}$-\sbmwis}\xspace}

\newcommand{\nph}{{\sf NP}-hard\xspace}
\newcommand{\npc}{{\sf NP}-complete\xspace}

\newcommand{\fpt}{{\sf FPT}\xspace}
\newcommand{\xp}{{\sf XP}\xspace}
\newcommand{\w}{{\sf W}\xspace}

\newcommand{\wone}{{\sf W}$[1]$-hard\xspace}

\newcommand{\Oh}{\ensuremath{\mathcal{O}}\xspace}

\newcommand{\A}{\ensuremath{\mathcal{A}}\xspace}

\newcommand{\I}{\ensuremath{\mathcal{I}}\xspace}

\newcommand{\Co}[1]{\ensuremath{\mathcal{#1}}\xspace}
\newcommand{\X}[1]{\ensuremath{\mathscr{#1}}\xspace}

\newcommand{\conflict}{conflict\xspace}
\newcommand{\object}{{job}\xspace}
\newcommand{\objects}{{jobs}\xspace}
\newcommand{\items}{{jobs}\xspace}

\newcommand{\assign}{assign\xspace}
\newcommand{\assigned}{assigned\xspace}

\newtheorem{defn}{Definition}

\newcommand{\yes}{yes}
\newcommand{\no}{no}

\newcommand{\sse}{\subseteq}

\newtheorem{clm}{Claim}
\newtheorem{obs}{Observation}
\newtheorem{remk}{Remark}

\newcommand{\Ma}[1]{{#1}}

\newcommand{\hide}[1]{}

\usepackage{comment}
\excludecomment{dontshow}

\usepackage[citecolor=red,linkcolor=blue]{hyperref}
\usepackage{cleveref}

\begin{document}
\title{How to assign volunteers to tasks compatibly ? A graph theoretic and parameterized approach}
%
%


\author{Sushmita Gupta\inst{1} \and
Pallavi Jain\inst{2} \and
Saket Saurabh\inst{1,3}}

\authorrunning{S. Gupta et al.}
%
\institute{Institute of Mathematical Sciences, HBNI, Chennai 
\email{\{sushmitagupta,saket\}@imsc.res.in}\\
\and
Indian Institute of Technology Jodhpur
\email{pallavi@iitj.ac.in}\\
\and
University of Bergen, Norway
}


%
\maketitle              
\begin{abstract}

In this paper we study a resource allocation problem that encodes correlation between items in terms of \conflict and maximizes the minimum utility of the agents under a conflict free allocation. Admittedly, the problem is computationally hard even under stringent restrictions because it encodes a variant of the {\sc Maximum Weight Independent Set} problem which is one of the canonical hard problems in both classical and parameterized complexity. Recently, this subject was explored by Chiarelli et al.~[Algorithmica'22] from the classical complexity perspective to draw the boundary between {\sf NP}-hardness and tractability for a constant number of agents. The problem was shown to be hard even for small constant number of agents and various other restrictions on the underlying graph. Notwithstanding this computational barrier, we notice that there are several parameters that are worth studying: number of agents, number of items, combinatorial structure that defines the conflict among the items, all of which could well be small under specific circumstancs. Our search rules out several parameters (even when taken together) and takes us towards a characterization of families of input instances that are amenable to polynomial time algorithms when the parameters are constant. In addition to this we give a superior $2^{m}|I|^{\Co{O}(1)}$ algorithm for our problem where $m$ denotes the number of items that significantly beats the exhaustive $\Oh(m^{m})$ algorithm by cleverly using ideas from FFT based fast polynomial multiplication; and we identify simple graph classes relevant to our problem's motivation that admit efficient algorithms.

\begin{dontshow}
In the standard model of fair allocation of resources to agents, every agent has some utility for every resource, and the goal is to \assign resources to agents so that the agents' welfare is maximised. Recent work in this domain considers the incompatibility between resources and \assign only mutually compatible resources to an agent. Motivated by this, we study the \cffafull (\cffa, in short) problem: given a set of agents, a set of resources, a utility function for every agent over a set of resources, and a \conflict graph on the set of resources; the goal is to \assign resources to the agents such that (i) the set of resources \assigned to an agent are compatible with each other, and the (ii) minimum satisfaction of an agent is maximised, where the satisfaction of an agent is the sum of the utility of the \assigned resources. 

\cffa is known to be \nph. We study the problem in the paradigm of parameterized complexity with respect to several natural parameters and obtain algorithmic results. Moreover, we study the problem on several classes of the \conflict graph, e.g., complete graphs and cluster graphs (a collection of cliques). The problem is polynomial-time solvable for a complete graph while \nph for a cluster graph. We also consider the parameter ``distance from tractability''. Towards this, we define some natural distance parameters from a complete graph and exhibit some algorithmic results.

\end{dontshow}

\end{abstract}

\keywords{Conflict free allocation \and fair allocation \and job scheduling \and independent set \and parameterized complexity.}

\section{Introduction} 

Imagine a situation where we are running a non-profit organization that specialises in volunteer work. Specifically, our objective is to bundle the tasks that need to be completed and pair them with 
the available volunteer workers in some meaningful way. Naturally, the volunteer workers have some preference over the available tasks and the tasks may have some inherent compatibility issues in that a person may only be assigned to at most one of the tasks that are mutually incompatible. The incompatibility among the tasks could be due to something as simple as the time interval in which they have to be performed. While it would be ideal to assign all the tasks, it may not actually be possible due to the above compatibility issues and the number of available workers. Moreover, this being a volunteer operation, the workers are "paid" by the satisfaction they derive from completing the bundle of tasks assigned to them. Thus, we want to ensure that the assignment is done in way that gives every volunteer worker the highest level of satisfaction possible. This is the setting of the job assignment problem studied in this article. 

The above described scenario falls under the more general topic of resource allocation which is a central topic  in economics and computation. Resource allocation is an umbrella term that captures a plethora of well-known problem settings where resources are matched to agents in a meaningful way that respects the preferences/choices of agents, and when relevant, resources as well.  Stable matching, generalized assignment, fair division, are some well-known problems that fall under the purview of resource allocation. These topics are extensively studied in economics, (computational) social choice theory, game theory, and computer science, to name a few; and are incredibly versatile and adaptable to a wide variety of terminology, techniques and traditions.

A well-known framework within which resource allocation is studied is in the world of {\sc Job Scheduling} problems on non-identical machines. In this scenario, the machines are acting as agents and the jobs are the tasks such that certain machines are better suited for some jobs than others and this variation is captured by the "satisfaction level" of the machine towards the assigned jobs. Moreover, the jobs have specific time intervals within which they have to be performed and only one job can be scheduled on a machine at a time. Thus, the subset of jobs assigned to a single machine must respect these constraints, and the objective can be both maximization and minimization as well as to simply test feasibility.  Results on the computational aspect of resource allocation that  incorporate  interactions and dependencies between the resources is relatively few. This is the backdrop of our work in this article.  A rather inexhaustive but representative list of papers that take a combinatorial approach in analysing a resource allocation problem and are aligned with our work in this paper is \cite{cheng2022restricted,bezakova2005allocating,kurokawa2018fair,DBLP:conf/atal/EbadianP022,DBLP:conf/atal/BarmanV21,ahmadian2021four,DBLP:conf/iwoca/ChiarelliKMPPS20,DBLP:journals/sigecom/Suksompong21,woeginger1997polynomial}. \Ma{In particular, we can point to the decades old work of Deuermeyer et. al~\cite{deuermeyer1982scheduling} that studies a variant of {\sc Job Scheduling} in which they goal is to assign a set of independent jobs to identical machines in order to maximize the minimal completion time of the jobs. Their {\sf NP}-hardness result for two machines (i.e two agents in our setting) is an early work with similar flavor. They analyse a well-known heuristic called the LPT-algoirthm to capture best-case performance and show that its worst case performance is 4/3-factor removed from optimum.  The more recent work of Chiarelli et. al~\cite{DBLP:conf/iwoca/ChiarelliKMPPS20} that studies "fair allocation" of indivisible items into pairwise disjoint subsets items that maximimizes the minimum satisfaction of the agents is the work that is closest to ours. They too consider various graph classes that capture compatibilities among items and explore the classical complexity boundary between strong {\sf NP}-hardness and pseudo-polynomial tractability for a constant number of agents. Our analysis probes beyond the {\sf NP}-hardness of these problems and explores this world from the lens of parameterized complexity, thereby drawing out the suitability of natural parameters--such as the number of agents, the number of \items, the maximum size of each allocated "bundle", and the structural parameters of the underlying graph--towards yielding polynomial time algorithms when the parameters take on constant values.}

We formally model our setting by viewing it is a two-sided matching market where each worker (i.e an {\it agent}) has a utility function defined over the set of available tasks (call them {\it \items}) such that their satisfaction for a bundle of \items is the sum of the agents' utilities for each individual \object in the bundle. The incompatibilities among the \items is captured by a graph \Co{H} defined on the set of \items such that an edge represents {\it conflict}. The overall objective is to assign bundles--pairwise disjoint subset of \items that induce an independent set in \Co{H} (i.e have no edges among each other)--to agents such that the minimum satisfaction of the agents is maximized. \hide{The \object graph is known as the {\it conflict graph} and the bundles that satisfy this constraint are called {\it conflict-free}.} To make our discussion concrete,  we formally define the computational  problem under study.

%

%
\defprob{\cffafull (\cffa)}{A set of agents \Co{A}, a set of \items \Co{I}, utility function $\util_a\colon \I \rightarrow \mathbb{N}$, for each agent $a\in \A$, a positive integer $\eta \in\mathbb{N}$; and a graph \Co{H} with vertex set \Co{I}.}{Does there exist a function  $\phi \colon \Co{A}  \rightarrow 2^{\Co{I}}$ such that  for every $a\in \Co{A}$, $\phi(a)$ is an independent set in \Co{H}, $\sum_{x\in \phi(a)}\util_a(x) \geq \eta$, and $\phi(a)\cap \phi(a') = \emptyset$ for all $\{a,a'\} \subseteq \A$?}
For each $a\in \Co{A}$, we call $\phi(a)$  a {\em bundle} assigned to the agent $a$. We call graph \Co{H}  the {\em conflict graph}.
\hide{In the literature, \cffa is called conflict free ``maximin fair'' allocation because we are trying to maximize the minimum utility (satisfaction) of agents.}

As alluded to earlier, Deurmeyer et al.~\cite{deuermeyer1982scheduling} studied \cffa with a different name and showed that the problem is \npc even for $2$ agents and even when \Co{H}  is edgeless (that is, no conflict). Since the problem statement has a conflict graph and we need the subsets of allocated resources to be an independent set in \Co{H}, naturally, the classical {\sc Maximum Weight Independent Set (MWIS)}  problem comes into play. 
In this problem given a graph $G$, a weight function $w: V(G)\rightarrow \mathbb{N}$, and an integer $\eta$, the objective is to test whether there exists an independent set $S$ such that $w(S)=\sum_{v\in S} w(v) \geq \eta$. Let $\mathcal G$ be a family of graphs.
Chiarelli et al.~\cite{DBLP:conf/iwoca/ChiarelliKMPPS20} showed that  if {\sc MWIS} is \npc on $\mathcal G$, then \cffa is \npc, when \Co{H} belongs to the graph class  $\mathcal G$, even when there is one agent.  Consequently, it is natural to focus on graph classes in which {\sc MWIS} is polynomial-time solvable.  However,  \cite{DBLP:conf/iwoca/ChiarelliKMPPS20} proves that \cffa remains \npc even for bipartite graphs and their {\em line graphs}. Some polynomial time algorithms  for special instances of the problem and polynomial time approximation algorithms are known for the problem~\cite{deuermeyer1982scheduling,hummel2022fair}. Some papers that have used conflict graphs to capture various constraints on items/jobs that are related to compatibility are~\cite{DARMANN20, Guy09,Factorovich20,Bouveret17}.

\subsection{Our Results and Methods} 


As described above we formulate our question in graph theoretic terms and analyze the problem in the realm of parameterized complexity. We note that this is a natural choice of algorithmic toolkit for our problem because \cffa {is naturally governed by several parameters such as the number of agents (${\sf \#agents}$), the number of \items (${\sf \#\items}$), the maximum size of a bundle (${\sf bundleSize}$) in the solution, and the utility of any agent $\eta$. This makes it  a natural candidate for a study from the viewpoint of parameterized complexity.  Moreover, we also note that for certain specific situations the \object graph may have special structures that can be exploited for designing efficient algorithms. In what follows, we describe certain scenarios where the ``small-ness" of the parameters and the underlying graph structure comes into focus and allows us to discuss our results more concretely.


\noindent {\bf Input/Output Parameters.}
 The first set of parameters, that we study, consists of 
$n={\sf \#agents}$, $m={\sf \#\items}$, $s={\sf bundleSize}$, and  $\eta$.
With this set of parameters, we obtain the following set of results. 

\smallskip
\noindent{\bf Brief overview of parameterized complexity:} The goal of parameterized complexity is to find ways of solving \nph problems more efficiently than exhaustive search: the aim is to
restrict the combinatorial explosion to a parameter that is likely to 
much smaller than the input size in families of input instances. Formally, a {\em parameterization}
of a problem is assigning an integer $k$  to each input instance of the problem.  We
say that a parameterized problem is {\em fixed-parameter tractable
  }(\fpt)  if there is an algorithm that solves the problem in time
$f(k)\cdot |I|^{\Oh(1)}$, where $|I|$ is the size of the input and $f$ is an
arbitrary computable function depending on the parameter $k$
only. A more general class of parameterized algorithm is the XP algorithms where a parameterized problem is {\em slicewise poly
  }(\xp)  if there is an algorithm that solves the problem in time
$|I|^{f(k)}$, where $|I|$ is the size of the input and $f$ is an
arbitrary computable function depending on the parameter $k$
only. Moreover, we will refer to such algorithms as an \fpt  (resp. \xp) algorithm and the problem to have an $\fpt(k)$ (resp. $\xp(k)$) algorithm. For more details on the subject, we refer to the textbooks~\cite{ParamAlgorithms15b,fg,downey}.

\noindent{\bf Graph classes under investigation:} We begin our discussion describing the simple graph classes that bookend our parameterized study: the two extremes are {\it utopia} and {\it chaos} and in between we have potentially an infinite possibility of graph classes in which to study our problem. In \Cref{ss:closer-look} we delve deeper into what parameters are meaningful for further study and draw out the connections between the graph classes and fruitful parameterization.


\begin{enumerate}
\item {\bf Utopia}: when there are no incompatibilities, and the \conflict graph \Co{H} is edgeless. In this scenario the problem is hard even when bundle size is a small constant, \Cref{thm:nph-edgeless}.

\item {\bf Chaos}: when every \object is incompatible with every other \object, and so \conflict graph \Co{H} is complete. In this scenario, the problem becomes rather easy to solve since each bundle can only be of size at most one, \Cref{thm:poly-clique}.

\item {\bf Incompatibilities are highly localized}: \Co{H} is a {\it cluster graph}, a graph that is compromised of vertex disjoint cliques. Such a situation may occur quite naturally such as in the following scenario.

In the example of the assignment of volunteers to tasks, consider the situation where the tasks can only be completed on specific days and specific times. Consequently, all the tasks that can be completed on day 1 form a clique, the ones for day 2 form another clique and so on. Moreover, the volunteers are working after hours for say two hours each day and it has been decided that each worker can only work for the same number of hours each day to manage their work load.  In this scenario a worker can be assigned at most task per day. This is the intuitive basis for the algorithm described in \Cref{thm:hardnessForcluster,thm:2cliques} and \Cref{obs:cluster-2agent}.

\item {\bf ``distance'' $t$ away from chaos}: \Co{H} has at least ${m \choose 2} -t$ edges, \Cref{thm:fpt-t+n,thm:fpt-t,thm:polytime-n-2}.

\end{enumerate}



%
\hide{Additionally, from the mathematical perspective, the parameters studied in this paper are natural input/output parameters that arise in the paradigm of parameterized complexity. }

If not a constant, it is reasonable to expect these parameters to be fairly small compared to the input.  

 

\subsection{Closer look at the parameters and search for fruitful graph families}\label{ss:closer-look}

\begin{description}[wide=0pt]

\item[I. (Superior) $\fpt(m)$ algorithm exists:] We note that  \cffa admits a trivial \fpt algorithm parameterized by $m$ by enumerating all possible $(n+1)^{m}$ ways of assigning the \items to the agents, where each \object has $(n+1)$ choices of agents to choose from. Since $m\geq n$, we get a running time of $\Oh(m^m)$. However, in {\bf \Cref{sec:fpt-input}} we present an algorithm with running time $2^{m}(n+m)^{\Oh(1)}$, which is clearly far superior. It is an algebraic algorithm that recasts the problem as that of 
of polynomial multiplication that mimics subset convolution. This suggests, in contrast to no $\fpt(n)$ algorithm, that  the larger parameter $m$ is sufficient in constraining the (exponential growth in the) time complexity as function of itself.
\smallskip

\item[II. No $\xp(n)$ algorithm exists:] We first note that since \cffa is \npc even for one agent (due to the reduction from {\sc MWIS} by Chiarelli et. al~\cite{DBLP:conf/iwoca/ChiarelliKMPPS20}), hence, we cannot even hope for an $(n+m)^{f(n)}$ time algorithm for any function $f$, unless {\sf P=NP}. Thus, there is no hope for an \fpt algorithm with respect to $n$. This appears to be a confirmation that the number of agents (volunteers) which is likely to smaller than the number of \items (tasks) is inadequate in terms of expressing the (exponential growth in the) time complexity as a function of itself.

\smallskip

\item[III. No $\xp(s)$ algorithm when \Co{H} is edgeless:] In {\bf \Cref{sec:agents+bundlesize}} we show that \cffa is \npc when \Co{H} is edgeless and $s=3$. This implies that we cannot even hope for an $(n+m)^{g(s)}$ time algorithm for any function $g$, unless {\sf P=NP}.  

Therefore, $n$ and $s$ are inadequate parameters individually, hence it is natural to consider them together.




\item[IV. When both $n$ and $s$ are small:] We note that $n$ and $s$ being small valued compared to $m$ is quite realistic because there are likely to be far too many tasks at hand but relatively fewer volunteers; and the assignment should not overburden any of them and thus the number of assigned tasks should be small.  This motivates us to consider the parameter $n+s$.

However, hoping that \cffa is \fpt parameterized by $n+s$ in general graphs is futile because the problem generalizes the {\sc MWIS} problem. Hence, we can only expect to obtain an $\fpt(n+s)$ algorithm for special classes of graphs. Consequently, our exploration moves towards identifying graph classes which may admit such an algorithm.

Towards that we note that an $\fpt(n+s)$ algorithm for the underlying decision problem that incorporates the bundle size $s$ (defined formally below) yields an $\fpt(n+s)$ algorithm for \cffa.  


 \defprob{\dcffafull (\dcffa)}{A set of agents \Co{A}, a set of \items~ \Co{I}, utility function $\util_a\colon \I \rightarrow \mathbb{N}$, for each agent $a\in \A$, positive integers $s,\eta \in \mathbb{Z}_{>0}$, and a graph \Co{H} with vertex set \Co{I}.}{Does there exist a function  $\phi \colon \Co{A}  \rightarrow 2^{\Co{I}}$ such that  for every agent $a\in \Co{A}$, bundle $\phi(a)$ is an independent set in \Co{H}, $|\phi(a)|\leq s$, $\sum_{x\in \phi(a)}\util_a(x) \geq \eta$, and $\phi(a)\cap \phi(a') = \emptyset$ for all $\{a,a'\} \subseteq \A$?}

To elaborate further, an \fpt algorithm for \dcffa would imply an \fpt algorithm for \cffa, because $s \leq m$ and an algorithm that makes subroutine calls to an $\fpt(n+s)$ algorithm for \dcffa for increasing values of $s$ (from $1$ to $m$) is an $\fpt(n+s)$ algorithm for \cffa.  Hence, from now we focus our attention towards an $\fpt(n+s)$ algorithm for \dcffa for reasonable graph classes.



$\star$ {\it Two parameters $n$, $s$ in search of a graph family:} A closer look at the hardness proof of \cffa from {\sc MWIS} by Chiarelli et. al~\cite{DBLP:conf/iwoca/ChiarelliKMPPS20} yields a hardness result for \dcffa from a size bounded version of {\sc MWIS}, defined below. Note that in this problem the size of the (independent set) solution is upper bounded by the parameter and this distinguishes it from the (standard) maximum weight independent set solution.


\defparprob{\sbmwisfull (\sbmwis) }{A  graph $G$, positive integers $k$ and $\rho$, a weight function $w: V(G) \rightarrow \mathbb{N}$.}{$k$}{Does there exist an independent set $S$ of size at most $k$ such that 
$\sum_{v\in S}w(v)\geq \rho$? }

In that reduction, $s=k$ and $n=1$. Hence, we can deduce the following connection: {\it any} $\fpt(n+s)$ algorithm for \dcffa will yield an $\fpt(k)$ algorithm for \sbmwis; and conversely, any hardness result that holds for \sbmwis with respect to $k$ must also hold for \dcffa with respect to $n+s$. The latter condition allows us to narrow down the potential graph classes that admit an $\fpt(n+s)$ algorithm for \dcffa. Since \sbmwis is a clear generalization (by setting $\rho=k$ and unit weight function) of the {\sc Independent Set} problem, a very well-studied problem in the realm of parameterized complexity and indeed the wider field of graph algorithms. This connection allows us to demarcate the tractability border of our problem \dcffa via the computational landscape of {\sc Independent Set}. In the paragraphs to follow we will flesh out the connection in more explicit terms and derive a result, {\bf \Cref{thm:equiavlence}}, that tightly characterises the tractability boundary of \dcffa with respect to $n+s$.







$\star$ {\it {\sc Independent Set} as the guiding light for \dcffa:} In the field of parameterized complexity  {\sc Independent Set}  has been extensively studied on families of graphs that satisfy some structural properties. We take the same exploration path for our problem \dcffa. The graph classes in which {\sc Independent Set} has an $\fpt(k)$ algorithm is a potential field for an \fpt algorithms for \dcffa. While this is not a guarantee, and we need to argue the connection precisely.

Let $\mathcal G$ be a family of {\em hereditary} graphs. That is, if $G\in {\mathcal G} $, then all induced subgraphs of $G$ belong to $\mathcal G$. In other words, $\mathcal G$ is closed under taking induced subgraphs.



For a hereditary family \Co{G}, \mwisg denotes the restriction of \mwis where the input graph $G\in \mathcal G$. Thus, a natural question is what happens when \dcffa is restricted to a graph class for which  \mwisg is fixed parameter tractable with respect to $k$? Given a hereditary family \Co{G}, we define \cffag similarly to \dcffa such that the graph \Co{H} belongs to the family \Co{G}. Tthe tractability of \mwisg does not immediately imply tractability of \cffag. Indeed, even if \mwisg is \fpt when parameterized by $k$, we cannot hope for an $(n+m)^{f(n)}$ time algorithm for \cffag for any function $f$, unless {\sf P=NP}, because the {\sf NP}-hardness of {\sc Independent Set} implies the {\sf NP}-hardness for \cffa even for one agent, i.e $n=1$.  Due to \Cref{thm:hardnessForcluster} (explained later), we also cannot hope for an $(n+m)^{f(s)}$ time algorithm for \cffag for any function $f$, unless {\sf P=NP}, even if  \mwisg has an $\fpt(k)$ algorithm. These 
 results imply that we cannot even expect \cffag to have an \xp algorithm with respect to either $n$ or $s$ individually, let alone an \fpt algorithm.  
 However, the following result completely characterizes the parameterized complexity of \cffag with respect to $n+s$ vis-a-vis the parameterized complexity of \mwisg with respect to $k$.



\begin{theorem}
\label{thm:equiavlence}
Let $\mathcal G$  be a hereditary family of graphs. Then, \cffag is \fpt parameterized by $n+s$ if and only if  \mwisg is \fpt parameterized by $k$. 
\end{theorem}

Theorem~\ref{thm:equiavlence} implies that \cffag is \fpt when $ \mathcal G$ is the family of interval graphs, chordal graphs, perfect graphs, planar graphs, bipartite graphs, graphs of bounded degeneracy, to name a few, \cite{DBLP:journals/siamcomp/FominTV15}.
%

\smallskip

\noindent{\bf Overview of Theorem~\ref{thm:equiavlence}.} This is one of the main  algorithmic results of this article. The result is obtained by combining the classical color coding technique of Alon-Yuster-Zwick~\cite{alon1995color},  applied on the set of \items, together with a dynamic programming algorithm to find a ``colorful solution''. In the dynamic programming phase of the algorithm, we  invoke an $\fpt(k)$ algorithm for \mwisg. 


While there are papers that study hereditary graph classes to give \fpt algorithms for {\sc MWIS} (the standard maximum weighted independent set) problem \cite{Dabrowski12},  we are not aware of known classes of graphs for which \mwis (the size bounded variant of maximum weighted independent set problem) is \fpt parameterized by $k$. Hence, we first identify some such graph classes.

 We define an {\em independence friendly class} as follows. Let $f  \colon \mathbb{N} \to \mathbb{N}$ be a monotonically increasing function, that is, invertible. A  
graph class $\mathcal G$ is called {\em $f$-independence friendly class ($f$-ifc)} if $\mathcal G$ is hereditary and for every 
$G\in \mathcal G$ of size $n$ has an independent set of size  $f(n)$. Observe that the families of bipartite graphs, planar graphs, graphs of bounded degeneracy, graphs excluding some fixed clique as an induced subgraphs are $f$-independence friendly classes with appropriate function $f$. For example, for bipartite graphs $f(n)=\nicefrac{n}{2}$ and for $d$-degenerate graphs $f(n)=\nicefrac{n}{(d+1)}$. For  graphs excluding some fixed clique as an induced subgraph, we can obtain the desired $f$ by looking at Ramsey numbers. It is the minimum number of vertices, $n = R(r, s)$, such that all undirected simple graphs of order $n$, contain a clique of size $r$, or an independent set of size $s$. It is known to be upper bounded by $R(r, s) \leq \binom{r + s - 2}{r - 1}$~\cite{jukna2011extremal}. 
%
We prove the following result for \mwisg when $\Co{G}$ is $f$-ifc. 
\begin{theorem}\label{lem:mwis-degenerate}
Let $\mathcal G$ be an $f$-independence friendly class. Then, there exists an algorithm for \mwisg running in time $\Oh((f^{-1}(k))^k\cdot (n+m)^{\Oh(1)})$. 
\end{theorem}
We also give a polynomial-time algorithm for \mwisg when $\Co{G}$ is a cluster graph. In contrast, \cffa is \nph when the conflict graph is a cluster graph as proved in \Cref{thm:hardnessForcluster}.

Finally, we show that \dcffa  is  \wone with respect to $n+s+\eta$. We reduce it from the {\sc Independent Set} problem.  Given an instance $(G,k)$ of  the {\sc Independent Set}, we can construct an instance of \cffa with only one agent, \items as $V(G)$, unit utility function, $\Co{H}=G$, and $s=\eta=k$. Since {\sc Independent Set} is \wone~\cite{downey1995fixed},  we get the following. 

\begin{obs}\label{obs:wh-eta}
\dcffa is \wone with respect to $n+s+\eta$. 
\end{obs}

\end{description}

Next, we move to our next set of parameters.

\smallskip
 
\subsection{Structural Parameterization via graph classes.}\label{ss:graph-classes}
Our next set of results is motivated by the following result whose proof is in {\bf \Cref{sec:distance-param}}.

\begin{theorem}\label{thm:poly-clique}
There exists an algorithm that solves \cffa in  polynomial time when the \conflict graph is a complete graph.
\end{theorem}

Contrastingly, we show that when \conflict graph is edgeless, the problem is computationally hard even when bundles are size at most three, \Cref{thm:nph-edgeless}.

This result leads us to asking if what happens if incompatibilities are highly localized:  

\begin{tcolorbox}[boxsep=5pt,left=5pt,top=5pt,colback=green!5!white,colframe=gray!75!black]
Does \cffa admit a polynomial time  algorithm when $\Co{H}$ is a disjoint union of cliques? 
\end{tcolorbox}

We answer this question negatively by proving the following result, which is due to a reduction from \NumMatch. 


\begin{theorem}
\label{thm:hardnessForcluster}
\cffa is \npc even when \Co{H} is a cluster graph comprising of $3$ cliques.
\end{theorem}

Since, an edgeless graph is also a cluster graph, due to \cite{deuermeyer1982scheduling}, we have the following. 
\begin{proposition}\label{obs:cluster-2agent}
\cffa is \npc even for $2$ agents when \Co{H} is a cluster graph.
\end{proposition}

Next, we design a polynomial-time algorithm when a cluster graph contains $2$ cliques and the utility functions are \emph{uniform}, i.e., utility functions are the same for all the agents. In particular, we prove the following result.

\begin{theorem}\label{thm:2cliques}
There exists an algorithm that solves \cffa in polynomial time when the \conflict graph is a cluster graph comprising of $2$ cliques and the utility functions are uniform.
\end{theorem}

Proofs of \Cref{thm:hardnessForcluster,thm:2cliques} are in {\bf \Cref{ss:cluster-proof}}. In light of Theorem~\ref{thm:poly-clique}, the {\em distance of a graph}  from a complete graph is a natural parameter to study in parameterized complexity.  
The distance function can be defined in several ways. We define it as follows:  the number of edges, say $t$,  whose addition makes the graph a complete graph.  
We first show a result that gives a {\em subexponential time algorithm}  when the number of agents is constant.
\begin{theorem}\label{thm:fpt-t+n}
There exists an algorithm that solves  \cffa in $\Oh((2t \cdot 2^{2\sqrt{t}}+1)^n(n+m)^{\Oh(1)})$ time, where $t=\binom{m}{2}-|E(\Co{H})|$ denotes the number of edges when added to \Co{H} yields a complete graph.
\end{theorem}

\Cref{thm:fpt-t+n}  is obtained by showing that if a graph $G$ can be made into a clique by adding 
at most $t$ edges then the number of independent sets of size at least $2$ is upper bounded by $\Oh(2t \cdot 2^{2\sqrt{t}})$.  However, it is not an \fpt algorithm {\em parameterized by $t$ alone}. To show that the problem is \fpt parameterized by $t$, we obtain the following result. 

\begin{theorem}\label{thm:fpt-t}
There exists an algorithm that solves \cffa in $\Oh((2t)^{t+1}(n+m)^{\Oh(1)})$ time. 
\end{theorem}

In light of \Cref{thm:poly-clique}, we know that \cffa is polynomial-time solvable when every vertex has degree $m-1$. Next, we show that the problem is also polynomial-time solvable when every vertex has degree $m-2$ and the utility functions are uniform.

\begin{theorem}\label{thm:polytime-n-2}
There exists an algorithm that solves \cffa in polynomial time when every vertex in the \conflict graph has degree $m-2$ and the utility functions are uniform.
\end{theorem}

Proofs of \Cref{thm:fpt-t+n,thm:fpt-t,thm:polytime-n-2} are in {\bf \Cref{ss:distance-from-completion}}.Table~\ref{tab:results} summarises all our results. 


\definecolor{olivegreen}{rgb}{.8,1,.8}
\newcommand{\tick}{{\color{olivegreen}{\bf \checkmark}}}
\newcommand{\cross}{{\color{red}{\ensuremath{\times}}}}

\definecolor{LightCyan}{rgb}{0.88,.8,.9}

\definecolor{lightred}{rgb}{1,.8,.8}

\newcolumntype{a}{>{\columncolor{LightCyan}}c}
\newcolumntype{b}{>{\columncolor{LightCyan}}c}
\newcolumntype{d}{>{\columncolor{LightCyan}}c}
\newcolumntype{e}{>{\columncolor{LightCyan}}c}

\newcommand{\Cy}[1]{{\textcolor{LightCyan}{#1}}}

\newcommand{\dash}{\cellcolor{yellow}}

\begin{table*}[htp]
\begin{center}
\begin{tabular}{| c  |c| b |c| d | d |c|}
\hline
Utility Functions & \multicolumn{3}{c|}{Arbitrary}& \multicolumn{2}{c|}{Uniform} & Arbitrary \\
\hline
\diagbox[innerwidth=3.8cm]{Parameters\vspace{0em}}{\Co{G}} & \multirow{2}{*}{Arbitrary} &  \multirow{2}{*}{Complete} &  \multirow{2}{*}{Cluster} & \multirow{1}{*}{Cluster}& \multirow{1}{*}{ Regular}&\multirow{1}{*}{ \mwisg is } \\ 
 & & & & 2 cliques & (degree \small{$\!m- 2$}) & \fpt wrt $k$ \\
\hline
$n=${\sf \#agents} & \cellcolor{lightred} \cite{deuermeyer1982scheduling} & {\bf Thm.}~\ref{thm:poly-clique} & \cellcolor{lightred}{\bf Obs.}~\ref{obs:cluster-2agent}& {\bf Thm.}~\ref{thm:2cliques} &{\bf Thm.}~\ref{thm:polytime-n-2}& \cellcolor{lightred}\cite{deuermeyer1982scheduling}\\
$s=$ {\sf bundleSize} & \cellcolor{lightred}~{\bf Thm.}\,\ref{thm:nph-edgeless} & & \cellcolor{lightred}~ {\bf Thm.}\,\ref{thm:hardnessForcluster} & & &\cellcolor{lightred}{\bf Thm.} \ref{thm:nph-edgeless}\\
$\eta$ & \cellcolor{lightred} {\bf Obs.}~\ref{obs:wh-eta} & & {\bf ?} & & &{\bf ?}\\
${\sf \#agents} + {\sf bundleSize}$& \cellcolor{lightred} {\bf Obs.}~\ref{obs:wh-eta} & & \cellcolor{olivegreen}~\bf{Thm.}\ref{thm:equiavlence} & & & \cellcolor{olivegreen}~{\bf Thm.}\,\ref{thm:equiavlence} \\
${\sf \# agents} + {\sf bundleSize}+ \eta$ & \cellcolor{lightred} {\bf Obs.}~\ref{obs:wh-eta} & & \dash & & & \dash \\
$m={\sf \#\items}$ & \cellcolor{olivegreen} {\bf Thm.}\,\ref{thm:fpt-with-items} & & \dash & & &\dash \\
 $t= {m\choose 2} - |E(\Co{H})|$ & \cellcolor{olivegreen}~{\bf Thm.}\,\ref{thm:fpt-t} &  & \dash & & & \dash \\
\hline
\end{tabular}
\end{center}
\caption{Summary of our results of \cffa, where the conflict graph belong to the family $\Co{G}$. \Cy{Lavender} cells denote polynomial time complexity; {\color{olivegreen} green} cells and {\color{lightred} pink} cells denote that the 
problem is \fpt and \w-hard w.r.t. the parameter in col 1, respectively; white cells with {\bf ?} mark denote that the complexity is open; and {\color{yellow} yellow} cells denote that the respective parameters are not interesting as the problem is either \fpt w.r.t. smaller parameter or for more general graph class.}
\label{tab:results}
\end{table*}%

\section{\cffa: A single exponential \fpt parameterized by ${\sf \#\items}$}\label{sec:fpt-input}

%
\newcommand{\polyn}[2]{\ensuremath{p^{#1}_{#2}}\xspace}

In this section, we will prove that \cffa is \fpt when parameterized by the number of \items, $m$.
The algorithm will use the technique of polynomial multiplication and fast Fourier transformation. The idea is as follows. For every agent $i\in \A$, we first construct a family of bundles that can be assigned to the agent $i$ in an optimal solution. Let us denote this family by $\Co{F}_i$. Then, our goal is to find $n$ disjoint bundles, one from each set $\Co{F}_i$. To find these disjoint sets efficiently, we  use the technique of polynomial multiplication. 

Before we discuss our algorithm, we have to introduce some notations and terminologies.  Let $\I$ be a set of size $m$, then we can associate \I with $[m]$.  The {\em characteristic vector} of a subset $S\subseteq [m]$, denoted by $\chi(S)$, is an $m$-length vector whose $i^{\text {th}}$ bit is $1$ if and only if  $i \in S$.  Two binary strings of length $m$ are said to be disjoint if for each $i\in [m]$, the $i^{th}$ bits in the two strings are different.  The {\em Hamming weight} of a binary string $S$, denoted by $H(S)$, is defined to be the number of $1$s in the string $S$. A monomial $y^i$ is said to have Hamming weight $w$, if the degree $i$ when represented as a binary string has Hamming weight  
$w$. 
 
We begin with the following observation. 

\begin{obs}\label{obs:disjoint-binary-vectors}
Let $S_1$ and $S_2$ be two binary strings of same length. Let $S=S_1+S_2$. If $H(S)=H(S_1)+H(S_2)$, then $S_1$ and $S_2$ are disjoint binary vectors. 
\end{obs}

The following is due to Cygan et. al~\cite{DBLP:journals/tcs/CyganP10}.

\begin{proposition}\label{prop:disjoint-set} Let $S=S_{1} \cup S_{2}$, where $S_1$ and $S_2$ are two disjoint subsets of $[m]$. Then, $\chi(S)=\chi(S_1)+\chi(S_2)$ and $H(\chi(S))=H(\chi(S_1))+H(\chi(S_2))=|S_1|+|S_2|$.
\end{proposition}

Observation~\ref{obs:disjoint-binary-vectors} and Proposition~\ref{prop:disjoint-set} together yield the following. 

\begin{corollary}\label{cor:HW-disjointness}
Subsets $S_{1}, S_{2} \sse \I$ are disjoint if and only if Hamming weight of the monomial $x^{\chi(S_1)+\chi(S_2)}$ is $|S_{1}|+|S_{2}|$. 
\end{corollary}


The {\em Hamming projection} of a polynomial $p(y)$ to $h$, denoted by $H_{h}(p(y))$, is the sum of all the monomials of $p(y)$ which have Hamming weight $h$. We define the {\em representative polynomial} of $p(y)$, denoted by $\Co{R}(p(y))$, as the sum of all the monomials that have non-zero coefficient in $p(y)$ but have coefficient $1$ in $\Co{R}(p(y))$, i.e., it ignores the actual coefficients and only remembers whether the coefficient is non-zero. We say that a polynomial $p(y)$ {\it contains a monomial} $y^i$ if the coefficient of $y^{i}$ is non-zero. The zero polynomial is the one in which the coefficient of each monomial is $0$.

Now, we are ready to discuss our algorithm.

\begin{theorem}\label{thm:fpt-with-items}
 \cffa is solvable in $\Oh(2^{m}(n+m)^{\Oh(1)})$ time, where $m={\sf \#\items}$ and $n={\sf \#agents}$.
\end{theorem}

\begin{proof}Let $\X{J}=(\A, \I, \{\util_{a}\}_{a\in \A}, \Co{H}, \eta)$ denote an instance of \cffa. 
We start by defining a set family indexed by the agents. Let $\A=[n]$. For an agent $i\in \A$, let $\Co{F}_{i}$ contain each of subsets of \Co{I} that can be {\it feasibly} allocated to $i$ as a bundle. Specifically, a set $S \sse \Co{I}$ is in $\Co{F}_{i}$ if $S$ is an independent set in \Co{H} and the utility $\sum_{x\in S} \util_{i}(x) \geq \eta$. We define the {\it round} \hide{of a polynomial} inductively as follows. 

For round $1$ and a positive integer $s$, we define a polynomial
\[\polyn{1}{s}(y) = \sum_{S\in \Co{F}_{1},\\ |S|=s} y^{\chi(S)} \]

For round $i \in [n] \setminus \{1\}$, and a positive integer $s$, we define a polynomial by using the $\Co{R}(\cdot)$ operator 

\[ \polyn{i}{s}(y) = \sum_{\substack{S \in \Co{F}_{i} \\ s'=s-|S|} } \Co{R}\left(H_{s}\left(\polyn{i-1}{s'}(y) \times y^{\chi(S)}\right) \right)
\]

The algorithm returns ``yes'' if for any positive integer $s\in \mathbb{Z}_{\geq 0}$, 
$\polyn{n}{s}(y)$ is non-zero. In fact, any non-zero monomial in the polynomial ``represents'' a solution for the instance \X{J} such that we can find the bundle to assign to each agent $i \in \A$ by backtracking the process all the way to round 1. 

\noindent{\it Computing a solution (if it exists).} We assume that for some positive integer $s$, $\polyn{n}{s}(y)$ is a non-zero polynomial. Thus, it contains a non-zero monomial of the form $\polyn{n-1}{s'}(y) \times y^{\chi(S)}$, where $S\in \Co{F}_{n}$. Note that $\chi(S)$ describes the bundle assigned to agent $n$, the set $S$.  Since the monomial $\polyn{n-1}{s'}(y) \times y^{\chi(S)}$ exists in the polynomial $\polyn{n}{s}(y)$ after applying $H_{s}(\cdot)$ function, it must be that $\polyn{n-1}{s'}(y) = y^{\chi(S')}$ for some set $S' \sse \Co{I}$ such that $S' \cap S = \emptyset$. By recursively applying the same argument to the polynomial $\polyn{n-1}{s'}(y)$, we can obtain the bundles that are allocated to  the  agents $i=n-1,\ldots,1$. 

\begin{lemma}
The above algorithm returns ``yes'' if and only if \X{J} is a \yes-instance of \cffa.
\end{lemma}

\begin{proof}Suppose that \X{J} is a \yes-instance of \cffa. Then, there is an {\it assignment}, i.e., an injective function $\phi$ that maps \A to subsets of \I. For each agent $i\in \Co{A}$, we define $S_{i} = \phi(i)$. We begin with the following claim that enables us to conclude that the polynomial $\polyn{n}{s}(y)$, where $s=\sum_{i\in [n]}|S_i|$, contains the monomial $y^{\sum_{i\in [n]}\chi(S_i)}$. 

\begin{clm}\label{clm:fft-fwd}
For each $j\in [n]$, the polynomial $\polyn{j}{s}(y)$, where $s=\sum_{i\in [j]}|S_i|$, contains the monomial $y^{\sum_{i\in [j]}\chi(S_i)}$. 
\end{clm}

\begin{proof}
The proof is by induction on $j$. 

\noindent {\bf Base Case:} $j=1$. We first note that $S_1$ is in the family $\Co{F}_1$ as it is a feasible bundle for the agent $1$. Thus, due to the construction of the polynomial $\polyn{1}{s}(y)$, we know that $\polyn{1}{|S_1|}(y)$ contains the monomial $y^{\chi(S_1)}$. 

\noindent{\bf Induction Step:} Suppose that the claim is true for $j= j'-1$. We next prove it for $j=j'$. To construct the polynomial $\polyn{j'}{s}(y)$, where  $s=\sum_{i\in [j']}|S_i|$, we consider the multiplication of polynomial $\polyn{j'-1}{s'}(y)$, where $s'=\sum_{i\in [j'-1]}|S_i|$, and $y^{\chi(S_{j'})}$. Due to the inductive hypothesis, $\polyn{j'-1}{s'}(y)$, where $s'=\sum_{i\in [j'-1]}|S_i|$, contains the monomial $y^{\sum_{i\in [j'-1]}\chi(S_i)}$. Note that $S_{j'}$ is in the family $\Co{F}_{j'}$ as it is a feasible bundle for the agent $j'$. Since $S_{j'}$ is disjoint from $S_1\cup \ldots \cup S_{j'-1}$, due to \Cref{cor:HW-disjointness}, we can infer that $\polyn{j'}{s}(y)$, where  $s=\sum_{i\in [j']}|S_i|$, has the monomial $y^{\sum_{i\in [j']}\chi(S_i)}$. 
\end{proof}

\begin{sloppypar}
Due to Claim~\ref{clm:fft-fwd}, we can conclude that $\polyn{n}{s}(y)$, where $s=\sum_{i\in [n]}|S_i|$, contains the monomial $y^{\sum_{i\in [n]}\chi(S_i)}$. For the other direction, suppose that the algorithm returns ``yes''. Then, for some positive integer $s$, $\polyn{n}{s}(y)$ is a non-zero polynomial. We need to show that there exists pairwise disjoint  sets $S_1,\ldots,S_n$ such that $S_i \in \Co{F}_i$, where $i\in [n]$. This will give us an assignment function $\phi$, where $\phi(i)=S_i$. Since each $S \in \Co{F}_i$, where $i\in [n]$, is an independent set and $\sum_{x\in S}\util_{i}(x) \geq \eta$, $\phi$ is a feasible assignment. We next prove the following claim that enables us to conclude the existence of pairwise disjoint sets. 
\end{sloppypar}

\begin{clm}\label{clm:fft-reverse}
For each $j\in [n]$, if the polynomial $\polyn{j}{s}(y)$ is non-zero for some $s\in [m]$, then there exists $j$ pairwise disjoint  sets $S_1,\ldots,S_j$ such that $S_i \in \Co{F}_i$, where $i\in [j]$. 
\end{clm}

\begin{proof}
We prove it by induction on $j$. 

\noindent{\bf Base Case:} $j=1$. Suppose $\polyn{1}{s}(y)$ is non-zero for some $s\in [m]$. Then, it contains a monomial $y^{\chi(S)}$, where $S\in \Co{F}_1$. Thus, the claim is true. 

\noindent{\bf Induction Step:} Suppose that the claim is true for $j=j'-1$. We next prove it for $j=j'$. Suppose that $\polyn{j'}{s}(y)$ is non-zero for some $s\in [m]$. Then, it contains a monomial of the form $\polyn{j-1}{s'}(y) \times y^{\chi(S)}$, where $|S|=s-s'$ and $S\in \Co{F}_{j'}$. Due to induction hypothesis, since  $\polyn{j-1}{s'}(y)$ is a non-zero polynomial,  there exists $j'-1$ pairwise disjoint  sets $S_1,\ldots,S_{j'-1}$ such that $S_i \in \Co{F}_i$, where $i\in [j'-1]$. Furthermore, due to \Cref{cor:HW-disjointness}, we have that $S_{j'}$ is disjoint from $S_1\cup \ldots \cup S_{j'-1}$. Thus, we have $j'$ pairwise disjoint  sets $S_1,\ldots,S_{j'}$ such that $S_i \in \Co{F}_i$, where $i\in [j']$.  \qed
\end{proof}
This completes the proof. \qed
\end{proof}

To claim the running time, we use the following well-known result about polynomial multiplication.

\begin{proposition}[\cite{moenck1976practical}]\label{prop:polynomial-multiplication}
There exists an algorithm that multiplies two polynomials of degree $d$ in $\Oh(d \log d)$ time.
\end{proposition}

\begin{lemma}\label{lem:time-fpt-m}
This algorithm runs in $\Oh(2^{m}\cdot(n+m)^{\Oh(1)})$ time. 
\end{lemma}

\begin{proof}
In the algorithm, we first construct a family of feasible bundles for each agent $i\in \A$. Since we check all the subsets of $\I$, the constructions of families takes $\Oh(2^m\cdot (n+m)^{\Oh(1)})$ time. For $i=1$, we construct $m$ polynomials that contains $\Oh(2^m)$ terms. Thus, $p_s^1(y)$ can be constructed in $\Oh(2^m\cdot m)$ time. Then, we recursively construct polynomials by polynomial multiplication. Since every polynomial has degree at most $\Oh(2^m)$, due to \Cref{prop:polynomial-multiplication}, every polynomial multiplication takes $\Oh(2^m \cdot m)$ time. Hence, the algorithm runs in  $\Oh(2^{m}\cdot(n+m)^{\Oh(1)})$ time.  \qed

\end{proof}
Thus, the theorem is proved. \qed

\end{proof}

\section{\cffa:  Parameterized by ${\sf \#agents}$ and ${\sf bundleSize}$ }\label{sec:agents+bundlesize}
In this section, we study \cffa  parameterized by $n={\sf \#agents}$,  ${\sf bundleSize}$,  
and their combinations. We first show some hardness results and then complement it with our main algorithmic result. 

\subsection{{\sf NP-hardness} when conflict graph is edgeless and bundle size is bounded }

Since \cffa is \nph for all the graph classes for which {\sc MWIS} is \nph~\cite{DBLP:conf/iwoca/ChiarelliKMPPS20}, in this section, we first discuss the intractability of the problem for special classes of graph when {\sc  MWIS} can be solved in polynomial time. 
In particular, we show that the problem is \nph even when the conflict-graph is edgeless and size of every bundle is at most $3$, which is due to the reduction from the \partition problem. In the \partition problem, we are given a set $X$ of $3\tilde{m}$ elements, a bound $B \in \mathbb{Z}_+$, and a size $s(x) \in \mathbb{Z}_+$  for each $x\in X$ such that $\nicefrac{B}{4} < s(x) < \nicefrac{B}{2}$ and $\sum_{x\in X} s(x)=\tilde{m}B$. The goal is to decide whether there exists a partition of $X$ into $\tilde{m}$ disjoint sets $X_{1}, X_{2}, \ldots, X_{\tilde{m}}$ such that for each $1\leq i \leq \tilde{m}$, $\sum_{x\in X_i} s(x) = B$. Note that each $X_i$ must contain three elements from $X$.  
To the best of our ability, we could not find a citation  for this result and hence we have included it here for completeness. 

\begin{theorem}\label{thm:nph-edgeless}
\cffa is \npc when \Co{H} is edgeless and $s$ is three.
\end{theorem}


\begin{proof}\begin{sloppypar}
Given an instance $\X{J}=(X,B,\{s(x)\}_{x\in X})$ of \partition, we create an instance $\X{J'}=(\Co{A}, \Co{I}, \{\util_{a}\}_{a\in \Co{A}}, \Co{H}, \eta=2B)$ of \cffa, where $\Co{I}=X$ and $\Co{H}$ is an edgeless graph on the vertex set $\Co{I}$.  We define a set of agents $\Co{A}=\{a_{1}, \ldots, a_{\tilde{m}}\}$ and for each agent $a_{i}\in \Co{A}$, we define the utility function $\util_{a_{i}}(x) = B-s(x)$ for each \object $x\in \Co{I}$.  
%
%
The intuition behind this construction is that we want to create a {\it bundle} so that the utility derived by an agent from that bundle is at least $2B$, which will be attainable only if the bundle size is three. 
\end{sloppypar}
Next, we prove the correctness of the reduction. 

\begin{lemma}
$\X{J}$ is a \yes-instance of \partition if and only if $\X{J}'$ is a \yes-instance of \cffa.
\end{lemma}

\begin{proof} If $\X{J}$ is a \yes-instance of \partition, then there is a solution $X_1,\ldots,X_{\tilde{m}}$ that satisfies the desired properties, i.e., for each $1\leq i \leq \tilde{m}$, $\sum_{x\in X_i} s(x) = B$. Note that $\sum_{x\in X_i}u_{a_i}(x)=3B-B=2B$. Thus, the assignment function $\phi$, where $\phi(a_i)=X_i$, yields a solution for $\X{J}'$. 

For the other direction, let $\phi$ be a solution for $\X{J}'$. That is, for each agent $a\in \Co{A}$, $\phi(a)$ is the bundle assigned to the agent $a$. Thus, $\sum_{x \in \phi(a)} \util_{a}(x)\geq 2B$. 

We claim that for each agent $a \in \Co{A}$, the bundle size $|\phi(a)| =3$. If the size is at most two, then $\util_{a}(\phi(a))\leq 2B- \sum_{x \in \phi(a)} s(x) < 2B$, since $\phi(a)$ is non-empty and for each object $x\in \phi(a)$, $s(x)$ is positive by definition. This is a contradiction. Hence, the only possibility is that  for each agent $a \in \Co{A}$, $|\phi(a)| \geq 3$. If for some agent $a\in \Co{A}$, bundle $\phi(a)$ has more than three \items, then for some agent $a' \neq a$, bundle $\phi(a')$ will contain at most two \items, and thus will not attain the target. Hence, for each agent, the bundle size is exactly three. 

Next, we claim that for each agent the utility of its bundle is exactly $\eta$. Suppose that there is an agent $a\in \Co{A}$, such that utility of its bundle, $\sum_{x\in \phi(a)}\util_{a}(x) > 2B$. By definition, $\sum_{x\in \phi(a)}\util_{a}(x) = 3B - \sum_{x \in \phi(a)} s(x)$. Thus, it follows that $\sum_{x \in \phi(a)} s(x) <B$. Since $\sum_{x\in X}s(x) = \tilde{m}B$, it must be that $\sum_{x \in \Co{I}\setminus \phi(a)} s(x) > (\tilde{m}-1)B$. 
Moreover, each bundle has size exactly three, and $\Co{I}\setminus \phi(a)$ has $3(\tilde{m}-1)$ \items, so there must exist a bundle $\phi(a')$ for some agent $a' \neq a$ such that $\sum_{x \in \phi(a')} s(x) > B$, and so that agent's utility $\util_{a'}(\phi(a')) = 3B- \sum_{x \in \phi(a')} s(x) < 2B$. Hence, we have reached a contradiction. Thus, for every agent $a\in \Co{A}$, the utility of its bundle is exactly $2B$. 

We now note that we can form a solution for the instance $\X{J}$ of \partition by taking each of the three-set \items constituting each bundle. More specifically, for each $i \in [\tilde{m}]$, we define $X_{i}= \{x \colon x \in \phi(a_i)\}$. For each $i\in [\tilde{m}]$, since $\util_{a_i}(x) = B- s(x)$, we have
\begin{equation*}
\begin{split}
\sum_{x \in X_{i}} s(x) & = \sum_{x\in \phi(i)} (B-\util_{a_i}(x)) = 3B - \sum_{x\in \phi(i)} \util_{a_i}(x)\\
& = 3B - 2B=B
\end{split}
\end{equation*}

Hence, $\X{J}$ is a \yes-instance of \partition. \qed 
\end{proof}
Thus, the theorem is proved. \qed
\end{proof}

%
%

%

\subsection{Proof of Theorem~\ref{thm:equiavlence}} 
\begin{sloppypar}
In this section, we give the proof of  Theorem~\ref{thm:equiavlence}. 
Let $\X{J}=(\A, \I, \{\util_{a}\}_{a\in \A}, \Co{H}, s,\eta)$  be an instance of \cffag, and let $|{\mathscr J}|$ denote the size of the instance.  We first prove the first part of Theorem~\ref{thm:equiavlence}, which is the easier direction of the proof.  In particular, let $\mathbb{A}$ be an \fpt algorithm for  \cffag, running in time $f(n,s) \ptime$. Given an instance $(G,k,\rho,w)$ of \mwisg, we construct an instance $\X{J}=(\A, \I, \{\util_{a}\}_{a\in \A}, \Co{H},s, \eta)$   of \cffag as follows. The set of agents $\A$ has only one agent $a^\star$. Further, $\I=V(G)$, $\util_{a^\star}=w$,  $\Co{H}=G$, $s=k$, and $\eta=\rho$. It is easy to see that 
$(G,k,\rho,w)$ is a \yes-instance of \mwisg if and only if $\X{J}$  is a  \yes-instance of \cffag. Thus, by invoking algorithm $\mathbb{A}$ on instance $\X{J}$ of \cffag, we get an \fpt algorithm for \mwisg that runs in $f(k) \ptime$ time. This completes the proof in the forward direction. In the rest of the section, we prove the reverse direction of the proof. That is, given an \fpt algorithm for \mwisg, we design an \fpt algorithm for \cffag.  For ease of explanation, we first present a randomized algorithm which will be derandomized later using the known tool of {\em $(p,q)$-perfect hash family}~\cite{alon1995color,fomin2014efficient}. 
\end{sloppypar}

%

\subsubsection{Randomized Algorithm}\label{sec:randomized algo}
In this section, we design a randomized algorithm with the following specification. If the input, \X{J}, is a  \no-instance then the algorithm always returns ``no''. However,   if the input, \X{J}, is a  \yes-instance then the algorithm returns ``yes'' with probability at least $1/2$.

Throughout this section, we assume that we have been given a \yes-instance. This implies that 
  there exists a  hypothetical solution $\phi \colon \Co{A} \rightarrow 2^{\Co{I}}$. We define everything with respect to $\phi$. That is, $\phi \colon \Co{A}  \rightarrow 2^{\Co{I}}$ is an injective function satisfying all the requirements. Let $S=\phi(\Co{A})=\cup_{a\in \Co{A}}\phi(a)$, i.e., the set of \objects that are assigned to some agent.  Further, note that $|S|\leq ns$, as the size of each bundle is upper bounded by $s$. 
Our main idea is to first highlight all the \objects in the set $S$, that are assigned to some agent, using color coding.   
\begin{tcolorbox}[boxsep=5pt,left=5pt,top=5pt,colback=green!5!white,colframe=gray!75!black]
{\bf  Separation of \items:}  Color the vertices of  $\Co{H}$ uniformly and independently at random using $ns$ colors, say $\{1,\ldots,ns\}$. 
\end{tcolorbox}

The goal of the coloring is that ``with high probability'', we color the \objects assigned to agents in a solution using distinct colors. The following proposition bounds the success probability. 
%

\begin{proposition}{\rm \cite[Lemma 5.4]{ParamAlgorithms15b}}\label{prop:success-prob} Let $U$ be a universe and $X\subseteq U$. Let $\chi \colon U \rightarrow [|X|]$ be a function that colors  each element of $U$ with one of $|X|$ colors uniformly and independently at random. Then, the probability that the elements of $X$ are colored with pairwise distinct colors is at least $e^{-|X|}$.
\end{proposition}

Due to \Cref{prop:success-prob}, the coloring step of the algorithm colors the \objects in $\phi(\Co{A})$ using distinct colors with probability at least $e^{-ns}$. We call an assignment $\phi \colon \Co{A} \rightarrow 2^{\Co{I}}$ as {\em colorful} if every two \items $\{i, i'\} \in \phi(A)$ get distinct color. Moreover, for each $a$, $|\phi(a)|\leq s$.


Next, we find a {\em colorful} feasible assignment in the following lemma.  Further, let us assume that we have an \fpt algorithm, $\mathbb{B}$, for \mwisg  running in time $h(k)n^{\Oh(1)}$.

\begin{lemma}\label{lem:colorful_solution}
Let  $\X{J}=(\A, \I, \{\util_{a}\}_{a\in \A}, \Co{H}, s,\eta)$
be an instance of \cffag and $\chi \colon V(\Co{H}) \rightarrow [ns]$ be a coloring function. Then, there exists a dynamic programming algorithm that finds a colorful feasible assignment $\phi \colon \Co{A} \rightarrow 2^{\Co{I}}$ in $\Oh(3^{ns}\cdot h(s) \cdot (n+m)^{\Oh(1)})$ time, if it exists, otherwise, return ``no''.
\end{lemma}

\begin{proof}
Let $\mathsf{colors} = \{1,\ldots,ns\}$ be the set of colors and let $a_1,\ldots,a_n$ be an arbitrary ordering of the agents. We apply dynamic programming: for a non-empty set $S\subseteq \mathsf{colors}$ and $i\in [n]$, we define the table entry $T[i,S]$ as $1$ if there is a colorful feasible assignment of \objects (that are colored by the function $\chi$) using colors in $S$ to agents $\{a_1,\ldots,a_i\}$; otherwise it is $0$.  For an agent $a\in \Co{A}$ and $S\subseteq \mathsf{colors}$, let $\Co{H}_{a,S}$ be a vertex-weighted graph constructed as follows. Let $V_S$ be the subset of vertices in $\Co{H}$ that are colored using the colors in $S$. Then,  $\Co{H}_{a,S}=\Co{H}[V_S]$. 
The weight of every vertex $x \in \Co{H}_{a,S}$ is $\util_a(x)$. For a vertex-weighted graph $G$, let $\mathbb{I}(G)\in \{0,1\}$, where 
$\mathbb{I}(G)=1$ if there exists an independent set of size at most $s$ and weight at  least $\eta$ in $G$, otherwise $0$. We compute $\mathbb{I}(G)$ using algorithm $\mathbb{B}$.
%
 We compute the table entries as follows. 
 
\noindent{\bf Base Case: } For $i=1$ and non-empty set $S$, we compute as follows:
\begin{equation}\label{eq:base case}
T[1,S] = \mathbb{I}(\Co{H}_{a_1,S})
\end{equation}

\noindent{\bf Recursive Step: } For $i>1$ and non-empty set $S$, we compute as follows:
\begin{equation}\label{eq:recursion}
T[i,S] = \bigvee_{\emptyset \neq S' \subset S} T[i-1,S'] \wedge \mathbb{I}(\Co{H}_{a_i,S\setminus S'})
\end{equation}

We return ``yes'' if $T[n,S]=1$ for some $S\subseteq \mathsf{colors}$, otherwise ``no''. 
Next, we prove the correctness of the algorithm. Towards this, we prove the following result. 
\begin{clm}\label{clm:correctness-equations}
\Cref{eq:base case} and \Cref{eq:recursion} correctly compute $T[i,S]$, for each $i\in [n]$ and $\emptyset \neq S \subseteq \mathsf{colors}$. 
\end{clm}

\begin{proof}
We will prove it by induction on $i$. For $i=1$, we are looking for any feasible assignment of \objects colored using the colors in $S$ to the agent $a_1$. Thus, \Cref{eq:base case} computes $T[1,S]$ correctly due to the construction of the graph $\Co{H}_{a_1,S}$ and the correctness of algorithm $\mathbb{B}$.  

Now, consider the recursive step. For $i>1$ and $\emptyset \neq S \subseteq \mathsf{colors}$, we compute $T[i,S]$ using \Cref{eq:recursion}. We show that the recursive formula is correct. Suppose that  \Cref{eq:recursion} computes $T[i',S]$ correctly, for all $i'<i$ and   $\emptyset \neq S \subseteq \mathsf{colors}$. First, we show that $T[i,S]$ is at most the R.H.S. of  \Cref{eq:recursion}. If $T[i,S]=0$, then the claim trivially holds. Suppose that $T[i,S]=1$. Let $\psi$ be a colorful feasible assignment to agents $\{a_1,\ldots,a_i\}$ using \objects that are colored using colors in $S$. Let $S_j\subseteq S$ be the set of colors of \objects in $\psi(a_j)$, where $j\in [i]$.  Since $\psi(a_i)$ uses the colors from the set $S_i$ and $\sum_{x \in \psi(a_i)}\util_{a_i}(x) \geq \eta$, due to the construction of $\Co{H}_{a_i,S_i}$, we have that  $\mathbb{I}(\Co{H}_{a_i,S_i})=1$.   Consider the assignment  $\psi'=\psi\vert_{\{a_1,\ldots,a_{i-1}\}}$ (restrict the domain to $\{a_1,\ldots,a_{i-1}\}$). Since  $S_i$ is disjoint from $S_1\cup \ldots \cup S_{i-1}$ due to the definition of colorful assignment, $\psi'$ is a feasible assignment for the agents $\{a_1,\ldots,a_{i-1}\}$ such that the color of all the \objects in $\psi'(\{a_1,\ldots,a_{i-1}\})$ is in $S\setminus S_i$. Furthermore, since $\psi$ is colorful, $\psi'$ is also colorful. Hence, $T[i-1,S\setminus S_i]=1$ due to induction hypothesis. Hence, R.H.S. of \Cref{eq:recursion} is $1$. Thus, $T[i,S]$ is at most R.H.S. of  \Cref{eq:recursion}.

For the other direction, we show that $T[i,S]$ is at least R.H.S. of  \Cref{eq:recursion}. If R.H.S. is $0$,  then the claim trivially holds. Suppose R.H.S. is $1$. That is, there exists $S' \subseteq S$ such that $T[i-1,S'] = 1$ and $\mathbb{I}(\Co{H}_{a_i,S\setminus S'})=1$.  Let $\psi$ be a colorful feasible assignment to agents $\{a_1,\ldots,a_{i-1}\}$ using \objects that are colored using colors in $S'$. Since $\mathbb{I}(\Co{H}_{a_i,S\setminus S'})=1$, there exists a subset $X\subseteq V_{S\setminus S'}$ such that $\sum_{x \in X}\util_{a_i}(x) \geq \eta$.  Thus, construct an assignment $\psi'$ as follows: $\psi'(a)=\psi(a)$, if $a\in \{a_1,\ldots,a_{i-1}\}$ and $\psi'(a_i)=X$. Since $\psi'$ is a feasible assignment and $\mathbb{I}(\Co{H}_{a_i,S\setminus S'})=1$, $\psi$ is a feasible assignment. Furthermore, since $\psi$ is colorful and $\psi(\{a_1,\ldots,a_{i-1}\})$ only uses colors from the set $S'$, $\psi'$ is also colorful. Hence, $T[i,S]=1$.
\qed \end{proof}

%
Due to Claim~\ref{clm:correctness-equations}, $T[n,S]=1$  for some $S\subseteq \mathsf{colors}$  if and only if $\X{J}$ is a yes-instance of \cffag. This completes the proof of the lemma.
\qed \end{proof}


Due to \Cref{prop:success-prob} and \Cref{lem:colorful_solution}, we obtain an $\Oh(3^{ns}\cdot h(s) \cdot (n+m)^{\Oh(1)})$ time randomized algorithm for \cffag which succeeds with probability $e^{-ns}$. Thus, by repeating the algorithm independently $e^{ns}$ times, we obtain the following result.

\begin{sloppypar}
\begin{theorem}\label{thm:randomized_algo}
There exists a randomized algorithm that given an instance $\X{J}=(\A, \I, \{\util_{a}\}_{a\in \A}, \Co{H}, s,\eta)$ of \cffag either reports a failure or finds a feasible assignment in $\Oh((3e)^{ns}\cdot h(s) \cdot (n+m)^{\Oh(1)})$ time. Moreover, if the algorithm is given a yes-instance, the algorithm returns ``yes'' with  probability at least $1/2$, and if the algorithm is given a no-instance, the algorithm returns ``no'' with  probability $1$. 
\end{theorem}
\end{sloppypar}
\begin{proof}
Let $\X{J}=(\A, \I, \{\util_{a}\}_{a\in \A}, \Co{H}, s,\eta)$ be an instance of \cffag. We color the \objects uniformly at random with colors $[ns]$. Let $\chi\colon V(\Co{H}) \rightarrow [ns]$ be this coloring function. We run the algorithm in \Cref{lem:colorful_solution} on the instance $\X{J}$ with coloring function $\chi$. If the algorithm returns ``yes'', then we return ``yes''. Otherwise, we report failure.  

Let $\X{J}$ be a yes-instance of \cffag and $\phi$ be a hypothetical solution. Due to \Cref{prop:success-prob}, all the \objects in $\phi(\Co{A})$ are colored using distinct colors with probability at least $e^{-ns}$. Thus, the algorithm in \Cref{lem:colorful_solution} returns yes with probability at least $e^{-ns}$. Thus, to boost the success probability to a constant, we repeat the algorithm independently $e^{ns}$ times. Thus, the success probability is at least 
\begin{equation*}
1-\Big(1-\frac{1}{e^{ns}}\Big)^{ns} \geq 1-\frac{1}{e} \geq \frac{1}{2}
\end{equation*}

If the algorithm returns ``yes'', then clearly $\X{J}$ is a yes-instance of \cffag due to \Cref{lem:colorful_solution}. 
\qed \end{proof}

\subsubsection{Deterministic Algorithm}\label{sec:deterministic algo}
We derandomize the algorithm 
using $(p,q)$-perfect hash family to obtain a deterministic algorithm for our problem. 

\begin{defn}[$(p,q)$-perfect hash family]{\rm (\cite{alon1995color})}
For non-negative integers $p$ and $q$, a family of functions $f_1,\ldots,f_t$ from a universe $U$ of size $p$ to a universe of size $q$ is called a $(p,q)$-perfect hash family, if for any subset $S\subseteq U$ of size at most $q$, there exists $i\in [t]$ such that $f_i$ is injective on $S$. 
\end{defn}

We can construct a $(p,q)$-perfect hash family using the following result. 

\begin{proposition}[\cite{naor1995splitters,ParamAlgorithms15b}]\label{prop:hash family construction}
There is an algorithm that given $p,q \geq 1$ constructs a $(p,q)$-perfect hash family of size $e^qq^{\Oh(\log q)}\log p$ in time $e^qq^{\Oh(\log q)}p \log p$.
\end{proposition}

Let $\X{J}=(\A, \I, \{\util_{a}\}_{a\in \A}, \Co{H}, s,\eta)$ be an instance of \cffag. Instead of taking a random coloring $\chi$, we construct an $(m,ns)$-perfect hash family $\Co{F}$ using \Cref{prop:hash family construction}. Then, for each function $f\in \Co{F}$, we invoke the algorithm in \Cref{lem:colorful_solution} with the coloring function $\chi=f$. If there exists a feasible assignment $\phi \colon \Co{A} \rightarrow 2^{\Co{I}}$ such that $|\phi(a)|\leq s$, for all $a\in \Co{A}$, then there exists a function $f\in \Co{F}$ that is injective on $\phi(\Co{A})$, since $\Co{F}$ is an $(m,ns)$-perfect hash family. Consequently, due to \Cref{lem:colorful_solution}, the algorithm return ``yes''. Hence, we obtain the following deterministic algorithm.

\begin{theorem}\label{thm:fpt-ns}
There exists a deterministic algorithm for  \cffag running in time 
$\Oh((3e)^{ns}\cdot (ns)^{\log ns} \cdot h(s)\cdot (n+m)^{\Oh(1)})$.
\end{theorem}

Due to \Cref{thm:fpt-ns}, we can conclude the following. 

\begin{corollary}\label{cor1:fpt-ns}
If \mwisg is solvable in polynomial time, then there exists a deterministic algorithm for  \cffag running in time 
$\Oh((3e)^{ns}\cdot (ns)^{\log ns} \cdot (n+m)^{\Oh(1)})$. 
%
%
\end{corollary}

It is possible that {\sc MWIS} is polynomial-time solvable on $\mathcal G$, but  \mwisg is \npc, as {\em any} $k$-sized solution of \mwis need not satisfy the weight constraint in the \mwisg problem. However, when we use an algorithm 
$\mathbb{B}$ for  \mwisg  in our algorithm, we could have simply used an algorithm for {\sc MWIS}. Though this will not result in a size bound of $s$ on the size of an independent set of weight at least $\eta$ that we found, however, this is sufficient to solve \cffa. However, we need to use \mwisg, when  {\sc MWIS} is \npc and we wish to use \fpt algorithm with respect to $k$. 
%
%
Due to \Cref{thm:fpt-ns} and this observation, \cffa is \fpt when parameterized by $n+s$ for several graph classes, such as chordal graphs~\cite{golumbic2004algorithmic},  bipartite graphs~\cite{golumbic2004algorithmic}, $P_6$-free graphs~\cite{grzesik2022polynomial}, outerstring graph~\cite{keil2017algorithm}, and fork-free graph~\cite{lozin2008polynomial}.

\begin{remk}
Our algorithm for chordal graphs is an improvement over the known algorithm that runs in $\Oh(m^{n+2}(Q+1)^{2n})$ time, where $Q=\max_{a\in \Co{A}}\sum_{i\in \Co{I}}p_{a}(i)$~{\rm \cite{DBLP:conf/iwoca/ChiarelliKMPPS20}}.
\end{remk}

\subsection{\fpt Algorithms for \mwisg when $\mathcal G$ is $f$-ifc 
}
%
In this section, we prove \Cref{lem:mwis-degenerate}. Let $(G,k,\rho,w)$ be a given instance of \mwisg. Further, $\mathcal G$ is  $f$-ifc. 
Let ${\sf HighWeight} = \{v\in V(G) \colon w(v) \geq \nicefrac{\rho}{k}\}$. Note that if there exists an independent set of $G[{\sf HighWeight}]$ of size $k$, then it is the solution of our problem. Since, $G$ belongs to $f$-ifc, $G[{\sf HighWeight}]$ is also $f$-ifc. Thus, there exists an independent set in $G[{\sf HighWeight}]$ of size at least  $f(|{\sf HighWeight}|)$.  
If $f(|{\sf HighWeight}|) \geq k$, then there exists a desired solution. To find a solution we do as follow. Consider an arbitrary  set 
$X \subseteq {\sf HighWeight}$ of size $f^{-1}(k)$. The size of $X$  guarantees that the set $X$ also has a desired solution. Now  we enumerate  subsets of size $k$ of $X$ one by one, and check whether it is independent; and if independent return it. This concludes the proof. 
 Otherwise, $|{\sf HighWeight}|< f^{-1}(k)$.
Note that the solution contains at least one vertex of ${\sf HighWeight}$. Thus, we guess a vertex, say $v$, in the set ${\sf HighWeight}$ which is in the solution,  delete $v$ and its neighbors from $G$, and decrease $k$ by $1$. Repeat the algorithm on the instance $(G-N[v],k-1,\rho-w(v),w\lvert_{V(G-N[v])})$. 

Since the number of guesses at any step of the algorithm is at most $f^{-1}(k)$ and the algorithm repeats at most $k$ times, the running time of the algorithm is $\Oh((f^{-1}(k))^k\cdot(n+m)^{\Oh(1)})$.

\begin{corollary}\label{cor:fpt-k-mwisg}
There exists an algorithm that solves \mwisg in $\Oh((2k)^k\cdot(n+m)^{\Oh(1)})$, $\Oh((4 k^2)^k\cdot(n+m)^{\Oh(1)})$, $\Oh((4k)^k\cdot(n+m)^{\Oh(1)})$, $\Oh((dk+k)^k\cdot(n+m)^{\Oh(1)})$, $\Oh(R(\ell,k)^k\cdot(n+m)^{\Oh(1)})$  time, when $\Co{G}$ is a family of bipartite graphs, triangle free graphs, planar graphs, $d$-degenerate graphs, graphs excluding $K_\ell$ as an induced graphs, respectively. Here, $R(\ell,k)$ is an upper bound on Ramsey number. 
\end{corollary}

\subsection{A polynomial time algorithm for \mwisg when $\mathcal G$ is a cluster graph}
In this section, we design a polynomial time algorithm. Let $\X{J}=(G,k,\rho,w)$ be a given instance of \mwisg. From each clique, we pick a vertex of highest weight. Let $X$ be the set of these vertices. Let $S\subseteq X$ be a set of $k$ vertices of $X$ that has highest weight. We return ``yes'' if $w(S)\geq \eta$, otherwise, ``no''. Next, we argue the correctness of the algorithm. If we return ``yes'', then, clearly, $S$ is an independent set of size at most $k$ and weight at least $\eta$. In the other direction, suppose that $Z$ is a solution to $\X{J}$. Suppose that $Z$ picks elements from the cliques $C_1,C_2,\ldots,C_\ell$, $\ell\leq k$. If $Z$ does not pick highest weight vertex from $C_j$, for some $j\leq l$, then we can replace the vertex $v=Z\cap C_j$ with highest weight vertex of $C_j$ in $Z$, and it is still a solution. Note that if $S\cap C_j = \emptyset$, where $j\leq \ell$, then $S$ contains a vertex whose weight is at least the weight of $v=S\cap C_j$ due to the construction of $S$. Since $|S|\geq |Z|$, we have a unique vertex for every such $j$. Thus, $w(S)\geq w(Z)\geq \eta$, and hence, the algorithm returns ``yes''.

\section{Distance From Chaos: Examining possible structures of the conflict graph}\label{sec:distance-param}


Starting point of our results in this section is the polynomial-time algorithm for \cffa when \conflict graph is a complete graph (there is an edge between every pair of vertices). We give proof of Theorem~\ref{thm:poly-clique}. We begin with a simple observation,  which follows due to the fact that the maximum size of an independent set in clique is $1$. 
\begin{obs}\label{bundlesize-clique}
If the \conflict graph is a complete graph, then the bundle size is $1$. 
\end{obs}

\begin{proof}[of Theorem~\ref{thm:poly-clique}]
Let $\X{J}=(\A, \I, \{\util_{a}\}_{a\in \A}, \Co{H}, \eta)$  be a given instance of \cffa. In light of Observation~\ref{bundlesize-clique}, we construct the following auxiliary bipartite graph $G=(L,R)$ as follows: for every agent $a\in \Co{A}$, we add a vertex $a$ in $L$ and for every \object $x\in \Co{I}$, we add a vertex $x$ in $R$.  If $p_a(x) \geq \eta$, then add an edge $ax$ in $G$. Next, we find a maximum matching $M$ in $G$. If $M$ does not saturate $L$, then return ``no'', otherwise return the following function $\phi \colon \Co{A}\rightarrow 2^{\Co{I}}$, $\phi(a)=M(a)$. 

Next, we prove the correctness of the algorithm. Clearly, if we return the function $\phi$, then it is a solution as we have an edge $ax$ if and only if  $p_a(x) \geq \eta$ and $M$ is a matching in $G$ saturating $L$. Next, we prove that if $\X{J}$ is a yes-instance of \cffa, then the algorithm returns a function. Let $\phi$ be a solution to $\X{J}$. Since $p_a(\phi(a)) \geq \eta$, we have an edge $a\phi(a)$ in $G$. Thus, there exists a matching in $G$ saturating $L$. Hence, the maximum matching $M$ in $G$ saturates $L$ and the algorithm returns a function. 
\qed \end{proof}
\subsection{When conflict is highly localized: conflict graph is a cluster graph}\label{ss:cluster-proof}
%

As discussed in the Introduction, there can be scenarios, where the incompatibilities are highly localized in a way that the set of \items can be decomposed into these small chunks where there is incompatibilities between all \items in the same chunk and none between \items in different chunks. Such a scenario is captured by a cluster graph


We show that the problem is intractable for cluster graph  even when it consists of $3$ cliques. This is in contrast to \Cref{thm:poly-clique}. To show the {\sf NP}-hardness, we give a polynomial time reduction from {\sc Numerical 3-dimensional Matching} problem, which is known to be \nph~\cite{garey1979computers}. In 
\NumMatch\ problem, we are given three disjoint sets $X$, $Y$, and $Z$, each containing $\tilde{m}$ elements, a size $s(a) \in \mathbb{Z}_{+}$ 
for each element $a\in X\cup Y\cup Z$,  and a bound $B\in \mathbb{Z}_{+}$. The goal is to partition $X\cup Y\cup Z$ into $\tilde{m}$ disjoint sets $A_1,\ldots, A_{\tilde{m}}$ such that (i) each $A_{i}$, where $i\in [\tilde{m}]$, contains exactly one element from each of $X$, $Y$, and $Z$, and (ii) for each $i\in [\tilde{m}]$, $\sum_{a\in A_{i}}s(a) = B$. Note that it follows that $\sum_{i\in [\tilde{m}]}\sum_{a\in A_{i}}s(a) = \tilde{m}B$.  Next, we give the  desired reduction.

\begin{proof}[of Theorem~\ref{thm:hardnessForcluster}]  Given an instance $\X{J}=(X, Y, Z, \{s_a\}_{a\in X\cup Y \cup Z},B)$ of the \NumMatch problem, we create an instance $\X{J'}=(\Co{A}, \Co{I}, \{\util_{a}\}_{a\in \Co{A}}, \Co{H}, \eta=B)$ of \cffa, where $\Co{I}=X\cup Y\cup Z$ and \Co{H} is a cluster graph on the vertex set \Co{I} with induced cliques on the vertices in the set $X$, $Y$, and  $Z$. We define a set of agents $\Co{A}=\{a_{1}, \ldots, a_{\tilde{m}}\}$ and for each agent $a_{i}\in \Co{A}$, we define the utility function $\util_{a_{i}}(j) = s(j)$ for each \object $j\in \Co{I}$. 
\begin{clm} \X{J} is a \yes-instance of \NumMatch if and only if \X{J'} is a \yes-instance of \cffa

\end{clm}

\begin{proof}Suppose that \X{J} is a \yes-instance of \NumMatch. Then, there is a solution $A_{i}, i\in [\tilde{m}]$, that satisfies the desired properties. It follows then that for the agent $a_{i}\in \Co{A}$, we have the condition that $\sum_{j\in A_{i}}\util_{a_{i}}(j) = \sum_{j\in A_{i}}s(j) =B$. Thus, the assignment function $\phi$, where $\phi(a_i)=A_i$, 
yields a solution for \X{J'} as well due to the construction of $\Co{H}$. 

Conversely, suppose that we have a solution for \X{J'}, i.e., an assignment function $\phi$ for which $\phi(a_{i}) \cap \phi(a_{j}) = \emptyset$ for every $\{a_{i}, a_{j}\} \sse \Co{A}$, and for every $a_{i}\in \Co{A}$, $\phi(a_i)$ is an independent set and $\sum_{j \in \phi(a_{i})} \util_{a_{i}}(j) \geq B$. 

Suppose that there exists an agent $a_{i}$ whose bundle $\sum_{j \in \phi(a_{i})} \util_{a_{i}}(j) > B$, then taking all the $\tilde{m}$ bundles together we note that $\sum_{i\in [\tilde{m}]}\sum_{j \in \phi(a_{i})} \util_{a_{i}}(j) > \tilde{m}B$, contradiction to the definition of the problem. Hence, we know that for each agent $a_{i}\in \Co{A}$, $\sum_{j \in \phi(a_{i})} \util_{a_{i}}(j) =B$. Since there are $\tilde{m}$ agents, every \object is assigned to some agent. Furthermore, since $\phi(a_i)$ is an independent set in $\Co{H}$, for each $i\in [\tilde{m}]$, we know that $\phi(a_i)$ contains $3$ elements, one from each clique. Hence, setting $A_{i}= \phi(a_{i})$ gives a solution to \X{J}. 
\qed \end{proof}
This completes the proof. \qed \end{proof}
%
Next, we show that when the cluster graph has only two cliques and the utility functions are uniform, i.e., for any two agents $a,a'$, $u_a=u_{a'}$, then \cffa can be solved in polynomial time. In particular, we prove \Cref{thm:2cliques}. For arbitrary utility functions, the complexity is open. 
We begin by noting that due to the utility functions being uniform, every bundle is valued equally by every agent. This allows us to look at the problem purely from the perspective of partitioning the \items into $n$ bundles of size at most two.



\begin{proof}[of \Cref{thm:2cliques}]
Let $\X{J}=(\A, \I, \{\util_{a}\}_{a\in \A}, \Co{H}, \eta)$  be a given instance of \cffa. Since the utility functions are uniform, we skip the agent identification from the subscript of utility function, i.e., instead of writing $\util_{a}$ for the utility function of agent $a$, we will only use $\util$. 

We note that if there exists an \object $z$ such that $\util(z) \geq \eta$, then there exists a  solution that assign it to some agent. Since the utility functions are uniform, it can be assigned to any agent. Let $\Co{I}_{\sf HighUtility}\subseteq \Co{I}$ be the set of \items whose utility is at least $\eta$, i.e., $\Co{I}_{\sf HighUtility} =\{z\in \Co{I} \colon \util(z)\geq \eta\}$. Let $\Co{I}_{\sf LowUtility}= \Co{I}\setminus \Co{I}_{\sf HighUtility}$. If  $|\Co{I}_{\sf HighUtility}|  \geq n$, then every agent get an \object from the set  $\Co{I}_{\sf HighUtility}$, and it is a solution. Otherwise, there are $|\Co{A}|-|\Co{I}_{\sf HighUtility}|$ agents to whom we need to assign bundles of size two. Let ${\sf IS}$ denote the set of all independent sets of size two in $\Co{H}[\Co{I}_{\sf LowUtility}]$. Thus, ${\sf IS}$ has size at most $m^2$. 

Next, we construct a graph, denoted by  $\Co{\widehat{H}}$,  on the \items in $\Co{I}_{\sf LowUtility}$ where there is an edge between vertices $a$ and $b$ if $\{a, b\} \in {\sf IS}$ and 
$u(a) + u(b)\geq \eta$. In this graph we compute a maximum sized matching, denoted by \Co{M}. If its size is less than  $n-|\Co{I}_{\sf HighUtility}|$, then we return the answer ``no''. Otherwise, we return answer ``yes'' and create an assignment as follows: if $(a,b) \in \Co{M}$, then we have a bundle containing $\{a,b\}$. We create exactly $n - |\Co{I}_{\sf HighUtility}|$ such bundles of size two and discard the others. These bundles along with the singleton bundles from $\Co{I}_{\sf HighUtility}$ yield our assignment for $n$ agents.

Clearly, this graph has $m$ vertices and at most $m^2$ edges. Thus, the maximum matching can be found in polynomial time. Next, we prove the correctness of the algorithm.

\noindent{\bf Correctness:} If the algorithm returns an assignment of \items to the agents, then clearly, for every agent the utility from the bundle is at least $\eta$. Every bundle is also an independent set in $\Co{H}$. Moreover, if a bundle is of size one, then the singleton \object is clearly an element of the set  $\Co{I}_{\sf HighUtility}$; otherwise, the bundle represents an independent set of size two in {\sf IS} whose total utility is 
 at least $\eta$. There are $n$ bundles in total,  exactly $|\Co{I}_{\sf HighUtility}| $ bundles of size one and at least $n - |\Co{I}_{\sf HighUtility}|$ bundles of size two. 
 
 In the other direction, suppose that $\phi$ is a solution to $\X{J}$. Let $a$ be an agent whose bundle size is two and $\phi(a)$ contains at least one \object from $\Co{I}_{\sf HighUtility}$, say $z$. Update the assignment $\phi$ as follows: $\phi(a)=\{z\}$. Note that $\phi$ is still a solution to $\X{J}$.   Let $\Co{A}_{1} \subseteq \Co{A}$ be the set of agents such that for every agent $a\in \Co{A}_{1}$, $|\phi(a)|=1$, i.e., the bundle size assigned to every agent in $\Co{A}_{1}$ is $1$. Clearly, $\phi(\Co{A}_{1})\subseteq \Co{I}_{\sf HighUtility}$. Let ${\sf rem}= \Co{I}_{\sf HighUtility}\setminus \phi(\Co{A}_{1})$, the set of unassigned ``high value'' \items. Suppose that ${\sf rem}\neq \emptyset$. 

Let $\Co{A}'' \subseteq \Co{A}\setminus \Co{A}_{1}$ be a set of size $\min\{|\Co{A}\setminus \Co{A}_{1}|,|{\sf rem}|\}$. Let $\Co{A}'' =\{a_1,\ldots,a_\ell\}$ and ${\sf rem}=\{z_1,\ldots,z_q\}$, where clearly $\ell \leq q$. 
Update the assignment $\phi$ as follows: for every $i\in [\ell]$, $\phi(a_i)=\{z_i\}$. Clearly, $\phi$ is still a solution of $\X{J}$. 
We note that there are only two cases: either $\Co{A}=\Co{A}_{1}\cup \Co{A}''$ or $\Co{\tilde{A}}=\Co{A}\setminus (\Co{A}_{1}\cup \Co{A}'')$ is non-empty. 

If $\Co{A}=\Co{A}_{1}\cup \Co{A}''$, then we have that the disjoint union of $\phi(\Co{A}_{1}) \cup \phi({\sf rem})\subseteq  \Co{I}_{\sf HighUtility}$. In other words,  $|\Co{I}_{\sf HighUtility}| \geq n$, and  
so there exists a solution in which every bundle is of size one and contains an element from  $\Co{I}_{\sf HighUtility}$. 

Otherwise, let $\Co{\tilde{A}}=\Co{A}\setminus (\Co{A}_{1}\cup \Co{A}'')$. Clearly, each of the \items in $\Co{I}_{\sf HighUtility}$ are assigned to agents in $\Co{A}_{1}\cup \Co{A}''$ and subsets of \items in $\Co{I}_{\sf LowUtility}$ are assigned to agents in $\Co{\tilde{A}}$. In other words, there exist $|\Co{I}_{\sf HighUtility}|$ bundles of size one and $n- |\Co{I}_{\sf HighUtility}|$ bundles of size two. Specifically for the latter, we know that each of the bundles is an independent set, they are pairwise disjoint and the total utility within each bundle is at least $\eta$. Thus, the members of each bundle share an edge in the graph $\Co{\widehat{H}}$ and the bundles themselves form a matching in the graph. Thus, our algorithm that computes a maximum matching in $\Co{\widehat{H}}$ would find a matching of size at least $n- |\Co{I}_{\sf HighUtility}|$. Hence, given the construction of the assignment from such a matching, we can conclude that our algorithm would return an assignment with the desired properties.  
\qed
\end{proof}

\subsection{Distance from chaos: parameterization by \#missing edges from complete graph}\label{ss:distance-from-completion}



In this section, we will prove that \cffa is \fpt with respect to the parameter $t$, the number of edges missing from \Co{H} being a complete graph. Further, we will present a polynomial time algorithm when the degree of every vertex in $\Co{H}$ is $m-2$ (one less than the degree in complete graph) and the utility functions are uniform. 

We first show a result that gives a {\em subexponential time algorithm}  when the number of agents is constant. 


\begin{proof}[of Theorem~\ref{thm:fpt-t+n}] We observe that the complement graph \Co{H}, denoted by $\overline{\Co{H}}$, contains all the vertices of \Co{H} but $t$ edges only. Moreover, each clique in this graph constitutes a \conflict-free bundle in the instance of \cffa. Conversely, we claim that any \conflict-free bundle in the instance of \cffa must form a clique in $\overline{\Co{H}}$ since for every pair of \items $x_{1}, x_{2}$ in a bundle, there exists an edge in $\overline{\Co{H}}$. 

Thus, enumerating all possible cliques (not just maximal ones) in $\overline{\Co{H}}$ allows us to check for possible allocations to agents. To show that this is doable in the claimed time, we will count the number of cliques in $\overline{\Co{H}}$. Since $\overline{\Co{H}}$ has $t$ edges, there can be at most $2t$ vertices that are not isolated. Vertices that are isolated {\em constitute a clique of size $1$}, and are called {\it trivial cliques}. They are upper bounded by the number of \items ($m$), and will be counted separately. A clique is said to be {\em non-trivial} if it does not contain an isolated vertex. Next, we will upper bound the non-trivial cliques. Towards this, we first show that $\overline{\Co{H}}$ is a $2\sqrt{t}$-degenerate graph by a simple counting argument. Note that if there exists a subgraph $H$ with minimum degree at least $2\sqrt{t}$, then the graph must have more than $t$ edges. Let $H$ be the subgraph of $\Co{H}$ induced on the non-isolated vertices of $\Co{H}$. Since $H$ has at most $t$ edges, every subgraph of $H$ has a vertex of degree at most $2\sqrt{t}$. Thus, $H$ is a $2\sqrt{t}$-degenerate graph, and hence has a $2\sqrt{t}$-degeneracy sequence. 
%
%
%
%
%
%
%
%
%


Let $\Co{D}= v_{1}, \ldots, v_{2t}$ denote a $2\sqrt{t}$-degenerate degree sequence of $H$. 
Notice that for any $i\in [2t]$, $v_{i}$ has at most $2\sqrt{t}$ neighbors among $\{v_{j}\colon  j>i\}$. Consider the  $2\sqrt{t}$ neighbors of $v_{1}$ and among them there can be at most $2^{2\sqrt{t}}$ cliques and can be enumerated in time $\Oh(2^{2\sqrt{t}})$. By iterating over $v_{i}$, we can enumerate all the non-trivial cliques  in $\overline{\Co{H}}$ in $\Oh(2t \cdot 2^{2\sqrt{t}})$ time. Indeed, for a non-trivial clique $C$, if $v_i$ is the first vertex in $C$ with respect to $\Co{D}$, that is all other vertices in $C$ appear after $v_i$ in  $\Co{D}$, then $C$ is enumerated when we enumerate all the cliques with respect to $v_i$ in our process.  This implies that the number of independent sets in   \Co{H}  
is upper bounded by $\Oh(2t \cdot 2^{2\sqrt{t}} + m) $ and the number of independent sets of size at least $2$ in   \Co{H}   is upper bounded by  $\Oh(2t \cdot 2^{2\sqrt{t}} ) $. Let $\mathbb{I}_{\geq 2}$ denote the family of independent sets of \Co{H} that have size at least $2$ -- the family of non-trivial independent sets. 

Thus, one potential algorithm is as follows. We first guess which agents are assigned non-trivial independent sets and which independent set.  That is, for each agent $a\in \Co{A}$, we guess  an independent set $I_a\in \mathbb{I}_{\geq 2} \cup \gamma$ ($\gamma$ is just to capture that the agent will not get non-trivial bundle). Let $\Co{A}' \subseteq \Co{A}$ be the set of agents for whom the guess is not $\gamma$. Let $(\Co{A}', \{I_a\}_{a\in\Co{A}'}) $ denote the corresponding guess for the agents in $\Co{A}'$. 
We first check that the guess for $\Co{A}' $  is {\em correct}. Towards that we check that for each $a_1,a_2 \in \Co{A}'$, $I_{a_1}\cap I_{a_2}=\emptyset$ and for each $a \in \Co{A}'$, we have that $\sum_{i\in I_a}\util_{a}(i) \geq \eta$.  Since, $|\mathbb{I}_{\geq 2}|$ is upper bounded by $\Oh(2t \cdot 2^{2\sqrt{t}})$, the number of guess are upper bounded by  $\Oh((2t \cdot 2^{2\sqrt{t}}+1)^n)$. For each correct  guess $(\Co{A}', \{I_a\}_{a\in\Co{A}'}) $, we solve the remaining problem by invoking Theorem~\ref{thm:poly-clique}. Let $\Co{A}^\star=\Co{A}\setminus \Co{A}'$ and $\Co{I}^\star=\Co{I}\setminus (\bigcup_{a\in\Co{A}'} I_a)$.   Then, we apply Theorem~\ref{thm:poly-clique} on the following instance: $(\Co{A}^\star,\Co{I}^\star,(p_a)_{a\in \Co{A}^\star},\eta, \Co{H}[\Co{I}^\star])$, here $\Co{H}[\Co{I}^\star]$ is a clique. This implies that the total running time of the algorithm is upper bounded by  $\Oh((2t \cdot 2^{2\sqrt{t}}+1)^n(n+m)^{\Oh(1)})$. 
\end{proof}


 However, Theorem~\ref{thm:fpt-t+n} is not an \fpt algorithm {\em parameterized by $t$ alone}. In what follows we design such an algorithm.  
%
\begin{proof}[of Theorem~\ref{thm:fpt-t}]
Let $\X{J}=(\A, \I, \{\util_{a}\}_{a\in \A}, \Co{H}, \eta)$ be a given instance of \cffa.  
Let $V_{>1}$ be the set of vertices which are part of independent sets of size at least $2$. As argued in the proof of \Cref{thm:fpt-t+n}, $|V_{>1}|\leq 2t$. Thus, there are at most $t$ bundles that contains more than one \object.  We guess partition of the \objects in $V_{>1}$  into at most $t+1$ sets, $\mathsf{notLarge, Large_1, \ldots, Large_\ell}$, where $\ell \leq t$, such that each set is an independent set in $\Co{H}$. The set $\mathsf{notLarge}$ might be empty. This contains the set of \objects in $V_{>1}$ which will not be part of any bundle of size at least $2$. The size of $\mathsf{Large}_i$ is at least $2$, for every $i\in [\ell]$, and each $\mathsf{Large}_i$ will be assigned to distinct agents in the solution. Next, we construct a complete graph $\Co{H}'$ as follows. For each $\mathsf{Large}_i$, where $i\in [\ell]$, we have a vertex $\mathsf{Large}_i$ in $\Co{H}'$, and $\util'_a(\mathsf{Large}_i)=\sum_{x\in \mathsf{Large}_i}\util_a(x)$, where $a\in \A$. If a vertex $v \in \Co{H}$ does not belong to any $\mathsf{Large}_i$, where $i\in [\ell]$, then add the vertex $v$ to $\Co{H}'$, and $\util'_a(v)=\util_a(v)$. Let $\X{J}'=(\A, \I, \{\util'_{a}\}_{a\in \A}, \Co{H}', \eta)$ be the new instance of \cffa where $\Co{H}'$ is a complete graph. Using \Cref{thm:poly-clique}, we find the assignment of bundles to the agents for the instance $\X{J}'$, if it exists, and return ``yes''. If the algorithm does not find the assignment for any guessed partition, then we return ``no''. The running time follows from \Cref{thm:poly-clique} and the fact that there are at most $(2t)^{t+1}$ possible partitions. 

Next, we prove the correctness of the algorithm. Suppose that $\X{J}$ is a yes-instance of \cffa and $\phi$ be one of its solution. Let $\Co{B}=\{\phi(a)\colon a\in \A \text{ and } |\phi(a)|\geq 2\}$. Clearly, sets in $\Co{B}$ are disjoint subsets of $V_{>1}$. Let $\Co{B}=\{B_1,\ldots,B_\ell\}$. Let $X\subseteq V_{>1}$ contains all the \objects that do not belong to any set in $\Co{B}$.  Since we try all possible partitions of $V_{>1}$, we also tried $B_1,\ldots,B_\ell,X$. Without loss of generality, let $B_i$ is assigned to $a_i$ under $\phi$. Thus, in the graph $\Co{H}'$, there is a matching $M=\{a_iB_i \in i\in [\ell]\} \cup \{a\phi(a) \colon |\phi(a)|=1\}$ that saturates $L$ in the proof of \Cref{thm:poly-clique}. Thus, the algorithm returns ``yes''. The correctness of the other direction follows from the correctness of \Cref{thm:poly-clique} and the construction of the instance $\Co{J}'$.
\end{proof}

Next, we give our claimed polynomial-time algorithm. 

\begin{proof}[of \Cref{thm:polytime-n-2}]
The algorithm is same as in \Cref{thm:2cliques}. Here, the size of ${\sf IS}$ is bounded by $\nicefrac{n}{2}$. 
\end{proof}

\section{Outlook}

In this article, we studied conflict-free fair allocation problem under the paradigm of parameterized complexity with respect to several natural input parameters. We hope that this will lead to the new set of results for the problem. The following question eludes so far: (i) the computation complexity of \cffa when the cluster graph contains only $2$ cliques with arbitrary utility functions, (ii) the computational complexity when the degree of every vertex in the  conflict graph is $n-2$ with arbitrary utility functions. 

Another direction of research is to consider various other fairness notions known in the literature, such as envy-freeness, proportional fair-share, min-max fair-share, etc., under the conflict constraint.


 \bibliographystyle{plain}
 \bibliography{references}

%
\end{document}